 \pdfoutput=1
\RequirePackage{fix-cm}
\documentclass[smallextended]{svjour3}       
\smartqed
\usepackage[utf8]{inputenc}
\usepackage{amsmath,amsfonts,amssymb,enumerate,multicol,nccmath}
\usepackage{tikz}
\usetikzlibrary{quantikz}
\usepackage{multirow}
\usepackage{graphicx,booktabs}
\usepackage{siunitx}
\usepackage{tabularx}
\usepackage{pifont}
\usepackage{blindtext}
\usepackage{cite}
\usepackage{pdflscape}
\makeatletter
\DeclareMathOperator*{\Motimes}{\text{\raisebox{0.25ex}{\scalebox{0.8}{$\bigotimes$}}}}
\renewcommand{\fnum@figure}{Fig. \thefigure}
\newcommand{\xmark}{\text{\ding{55}}}
\makeatother

 \journalname{}
\begin{document}

\title{A novel three party Quantum secret sharing scheme based on Bell state sequential measurements with application in quantum image sharing}

\titlerunning{Quantum secret sharing }        

\author{Farhan Musanna        \and
        Sanjeev Kumar$^*$ 
}


\institute{Farhan Musanna \at
              Department of Mathematics, IIT Roorkee \\
              \email{fmusanna@ma.iitr.ac.in}           
           \and
           Corresponding Author: Sanjeev  Kumar \at
              Department of Mathematics, IIT Roorkee \\
              \email{malikfma@iitr.ac.in}
}
\date{Received: date / Accepted: date}

\maketitle

\begin{abstract}
	In this work, we present a quantum secret sharing scheme based on Bell state entanglement and sequential projection measurements. The protocol verifies the $n$ out of $n$ scheme and supports the aborting of the protocol in case all the parties do not divulge in their valid measurement outcomes. The operator-qubit pair forms an integral part of the scheme determining the classical secret to be shared. The protocol is robust enough to neutralize  any eavesdropping on a particular qubit of the dealer. The experimental demonstration of the scheme is done on IBM-QE cloud platform with backends \texttt{IBMQ\_16\_Melbourne} and \texttt{IBMQ\_QASM\_SIMULATOR\_V0.1.547} simulator. The security analysis performed on the scheme and the comparative analysis supports our claim of a stringent and an efficient scheme as compared to some  recent quantum and semi-quantum techniques of secret sharing.
\end{abstract}
\section{Introduction}
   Nowadays, technology has advanced by leaps and bounds. The world at present has seen a paradigm shift where technology is not an alternative but a necessity. Ranging from transferring high profile date content across the globe, making hefty payments, signing important documents, drawing a blueprint of a city, technology has come to the rescue. Therefore, people do not tend to keep their secrets in one place. Instead, they try to distribute them among many parties. This act of sharing a secret among many parties is known as secret sharing in the cryptographic parlance. It emerged as a technique to limit the power of an individual in having sole information about the secret. The widespread popularity of this technique was due to the seminal work by Shamir et al. in \cite{shamir1979share}. This formal idea of sharing information between parties caught the researchers' attention, and active research started in this direction.
   \section{Existing Works}
   In the last couple of decades, the researchers have gone a step ahead and conceived computations on a quantum computer. The results of this rigorous pursuit of exploring an entirely new paradigm are the main motivations behind the development of breakthrough algorithms like Shor's algorithm \cite{shor1994algorithms}, Deutsche algorithm \cite{deutsch1992rapid}, Grover search algorithm \cite{grover1997quantum}, No-cloning \cite{buvzek1996quantum} and Holevo theorems \cite{holevo1973bounds}. The realm of quantum cryptography covers a diverse range of concepts like Quantum Key Distribution (QKD) \cite{bennett2014quantum,ekert1991quantum,bennett1992quantum}, Identification systems \cite{duvsek1999quantum}, Quantum secure communications \cite{bostrom2002deterministic,deng2003two}, Quantum digital signature \cite{zeng2002arbitrated,lee2004arbitrated,li2009arbitrated}. In this context, quantum secret sharing is considered an important problem in the area of secure communications.
   \par The work proposed by Hillery et al. \cite{hillery1999quantum} provides a foundation for quantum secret sharing. They relieved the applicability of the entanglement property of the Greenberger-Horne-Zeilinger (GHZ) state to create shares for different colluding parties \cite{greenberger1989bell}. Their work primarily focused on the right choice of measurement basis used by the parties for secret reconstruction. The protocol could also detect intrusion by any malicious party using the properties of the GHZ state. Later, Karlsson et al. \cite{karlsson1999quantum} proposed a $(m,n)$ scheme based on the multi-particle entanglement measurements. One distinct feature of this scheme was the relaxation that $m\leq n$, i.e., even if $m$ parties out of the total $n$ agreed to collude, the secret would be constructed. One of the highlights of their scheme was the use of non-orthogonal entangled states to detect an eves dropper and prevent the integrity loss of the data. Xiao et al. \cite{xiao2004efficient} gave another refinement of the pioneer schemes by introducing the concept of Quantum Key Distribution (QKD) to increase the efficiency of the scheme.  The work by Zhang et al. \cite{zhang2005multiparty} used similar concepts of unitary transforms and multiparty measurement outcomes to reveal the secret without using entanglement.  Guo et al. \cite{guo2003quantum} implemented a secret sharing scheme independent of the entangled state by directly encoding the photon of QKD, thereby enhancing computational efficiency.  Some other works have been reported in the literature that ensures provably secure schemes for quantum secret sharing \cite{gottesman2000theory,tittel2001experimental,deng2005improving}. Experimental quantum sharing was reported by Schmidt et al. in \cite{schmid2005experimental}. The scheme had a back-door eavesdropping strategy, which could lead to data manipulations.

   Apart from the earlier described protocols, various schemes made use of the positioning of the photons and Bell basis measurement. Zhang and Man in \cite{zhang2005multiparty1} adopted the positioning-based scheme and devised an efficient scheme for classical secret sharing. Markham and Sanders \cite{markham2008graph} used a graph state approach for devising a $(3,5)$ sharing scheme. Hsu \cite{hsu2003quantum} proposed another novel idea in this direction, which harnessed the Grover search algorithm for secret sharing. Fortescue and Gour \cite{fortescue2012reducing} proposed an amalgamation of classical and quantum encryption to design a novel perfect quantum secret sharing based on imperfect ramp sharing. The novelty of this scheme lay in the reducing the communication cost of secret sharing based on this hybridization. They discussed the important theoretical aspects of the size of the shares and the number of participants required to reconstruct the secret. Mitra et al. \cite{maitra2015proposal} proposed a novel idea for a rational quantum secret sharing scheme in which they gave the theoretical concepts of utility-strategy based on Nash equilibrium, which restricts the scenario wherein an individual participant tries to reconstruct the secret exclusively and does not divulge in its share. A very recent quantum secret sharing was proposed in \cite{qin2020hierarchical} harnessed the properties of a higher-dimensional entangled state that facilitated the dynamic nature of the number of participants and the number of shares each participant holds in the scheme. 
   \par There are hybrid schemes that are an amalgam of classical and quantum techniques used in secret sharing. The scheme reported by Yang et al. in \cite{yang2013secret} used Quantum Fourier Transform (QFT) as the basis for generating the secret shares. Qin et al. \cite{qin2018multi} used QFT and generalized Pauli operators to design the shares for each participant and validated their scheme's correctness. The scheme presented by Xiao \cite{xiao2013multi} used QFT, generalized Pauli operator, and $(n+1)$ GHZ states to share the secret, requiring every party to do the QFT for recovering the secret, which could be practically inefficient as the number of participants increase. Song et al. \cite{song2017t} implemented the QFT, Pauli operators and C-Not gates for secret sharing of classical information using $d$-level quantum states. The scheme lacked some basic features like the recoverability of the secret and the theoretical aspect of QFT. Mashhadi \cite{mashhadi2019general} proposed a scheme that was primarily based on the Monotone Span Program to generate an Access structure and generate shares. The quantum part of the schemes was exclusive to the application of QFT, Pauli operators, and individual measurements of parties. Some of the generic steps in the schemes like\cite{yang2013secret,qin2018multi,xiao2013multi,song2017t,mashhadi2019general} are summarized as follows:
   \begin{enumerate}
   	\item Generate shares $s_1,s_2,....s_n$ exclusively by classical schemes like MSP, Lagrange interpolation etc.
   	\item Initialize a qubit and apply the QFT to get $\ket{\phi}=\sum_{x=0}^{d-1}\ket{x}\ket{0}....\ket{0}$.
   	\item Apply a C-NOT gate to get  $\ket{\phi'}=\sum_{x=0}^{d-1}\ket{x}\ket{x}....\ket{x}$.
   	\item Apply the Pauli operator $U_{s_i,\beta}=\sum_{j=0}^{d-1}\omega^{j\beta}\ket{j+s_i}\bra{j}$ on $\ket{\phi'}$.
   	\item Obtain the state: 
   	
   	$\ket{A}=\sum_{j_1,j_2..j_n=0}^{d-1} \ket{j_1+s_1} \ket{j_2+s_2}.... \ket{j_3+s_3}$.
   	\item Each party applies its own individual measurement to get their shares $(s_1,s_2,....,s_n)$ and use it to construct the secret $s$.
   \end{enumerate}
   \subsection{Motivations and Novelty}
   \par The manifestation of this article is a result of the effort put in to propose a secure quantum cryptosystem without the use of classical means whose incorporation in the scheme would have diluted the whole purpose. Some of the salient features of the proposed scheme are given as follows:
   \begin{itemize}
   	\item The dealer $D$ generates the Einstein Podolsky Rosen (EPR) pair. The state of EPR pair is known only to its generator, which increases the security of the protocol, as shown in subsequent sections.
   	\item The operations made by the dealer $D$ on the position of the qubits determines the classical secret to be shared. This secret sharing aspect was predominantly missing in the earlier schemes, where only the operator determined the secret.
   	\item The operated qubit is disclosed by the $D$ only after successive measurements by the parties.
   	\item The collapsed states after measurement of each party makes any outcome equally likely, thereby requiring mandatory participation of each one in the protocol.
   	\item There is a check on the EPR state against any eavesdropping. Therefore, reconstruction of the secret is not possible if there are any manipulations by the adversary.
   \end{itemize}
   \section{Proposed Secret Sharing Protocol}
   The proposed protocol describes a $(3,3)$ quantum secret sharing scheme. However,it can be generalized to a $(n,n)$ scheme, where all the share-holders need to collude for reconstructing the original secret $s$. This section is consisting of design and implementation of protocol along with the secret reconstruction strategies of the proposed scheme.
   \subsection{Protocol Design}
   We utilize the single qubit gates $\mathbb{I},\mathbb{Z},\mathbb{Y},\mathbb{X}$ to design our protocol. We consider that the participants in this scheme are $P_1,P_2,P_3$. Apart from these participant, we have a dealer $D$ who wishes to share the secret $s$ among these participants as a sequence of qubits. The dealer $D$ creates a product state with one of the possible Bell states out of the four states $\ket{\alpha^+}, \ket{\alpha^-},\ket{\beta^+},\ket{\beta^-}$. We assume the hypothesis that the dealer $D$ some time back has already shared a Bell state \begin{scriptsize}\begin{equation}\label{Bell}
   	\ket{\psi}=\left(\dfrac{\ket{0}_1\ket{0}_2+ \ket{1}_1\ket{1}_2}{\sqrt{2}}\right) \left(\dfrac{\ket{0}_3\ket{0}_4+ \ket{1}_3\ket{1}_4}{\sqrt{2}}\right)\left(\dfrac{\ket{0}_5\ket{0}_6+ \ket{1}_5\ket{1}_6}{\sqrt{2}}\right)
   	\end{equation}\end{scriptsize} with all the three participants. The dealer allocates a pair of qubits to each participant. Assume that $P_1$ gets qubits $1$ and $4$, $P_2$ gets qubits $2$ and $6$, $P_3$ gets qubits $3$ and $5$. This allocation can be any other combination provided each party has a share of the other entangled particle. The dealer $D$ has with him a set of four unitary operators $\mathbb{I}, \mathbb{X}, i\mathbb{Y},\mathbb{Z}$. Also, associated with each of these operators, a string of bits given by the following correspondence: \begin{eqnarray}\label{secret}
   \mathbb{I}_1 &=& 00, \mathbb{I}_4 = 11 \nonumber \\
   \mathbb{X}_1 &=& 01, \mathbb{X}_4 = 10 \nonumber\\
   i\mathbb{Y}_1 &=& 11, i\mathbb{Y}_4 = 00 \nonumber\\
   \mathbb{Z}_1&=& 10, \mathbb{Z}_4 = 01
   \end{eqnarray}
   where, the operator $\mathbb{Z}_1$ means that the operator $\mathbb{Z}$ is acting on qubit $1$, and similarly $\mathbb{Z}_4$ means that $\mathbb{Z}$ is  acting on qubit $4$. The protocol to share the secret $s$ as binary string of length $8$ is shared by repeating the protocol $4$ times. For instance the dealer wants to share $s=55=00110111$, the dealer will actually manipulate the state $\ket{\psi}$ with the operations $ (\mathbb{I}i \mathbb{Y}\mathbb{X}i\mathbb{Y})_1$.
   \subsection{Protocol Implementation}
   The overall implementation of the proposed protocol is carried out in the following nine steps:
   \begin{enumerate}
   	\item The dealer $D$ decides to alter the particles of $P_1$, i.e., $(1,4)$. It also decides to operate the $\mathbb{X}$ gate on particle $1$  thereby resulting in the state as follows: 
   	\begin{scriptsize}\begin{equation}\label{2}
   		\resizebox{1\hsize}{!}{$\ket{\psi'}=\left(\frac{\ket{1}_1\ket{0}_2+ \ket{0}_1\ket{1}_2}{\sqrt{2}}\right) \left(\frac{\ket{0}_3\ket{0}_4+ \ket{1}_3\ket{1}_4}{\sqrt{2}}\right)\left(\frac{\ket{0}_5\ket{0}_6+ \ket{1}_5\ket{1}_6}{\sqrt{2}}\right)$}
   		\end{equation}
   	\end{scriptsize}
   	\item The dealer $D$ gives the state $\ket{\psi'}$ to $P_1$ who decides to measure his share of the entangled particles, i.e., $(1,4)$ in the Bell basis $\ket{\alpha^+}_{14}, \ket{\alpha^-}_{14},\ket{\beta^+}_{14},\ket{\beta^-}_{14}$. The quantum measurement operators involved in $P_1's$ measurements are as follows:\begin{scriptsize}\begin{equation}
   		\begin{aligned}
   		\bra{\alpha^+}_{14}={}& \dfrac{1}{\sqrt 2}\Bigl(\bra{0} \otimes \mathbb{I} \otimes \mathbb{I} \otimes  \bra{0} \otimes \mathbb{I}  \otimes  \mathbb{I} + \bra{1} \otimes \mathbb{I} \otimes \mathbb{I} \otimes \bra{1}  \otimes  \mathbb{I}\otimes \mathbb{I}\,\Bigr)\\
   		\bra{\alpha^-}_{14}=& \dfrac{1}{\sqrt 2}\Bigl(\bra{0} \otimes \mathbb{I} \otimes \mathbb{I} \otimes  \bra{0} \otimes \mathbb{I}  \otimes  \mathbb{I} - \bra{1} \otimes \mathbb{I} \otimes \mathbb{I} \otimes \bra{1}  \otimes  \mathbb{I}\otimes \mathbb{I}\,\Bigr)\\
   		\bra{\beta^+}_{14}=& \dfrac{1}{\sqrt 2}\Bigl(\bra{0} \otimes \mathbb{I} \otimes \mathbb{I} \otimes  \bra{1} \otimes \mathbb{I}  \otimes  \mathbb{I} + \bra{1} \otimes \mathbb{I} \otimes \mathbb{I} \otimes \bra{0}  \otimes  \mathbb{I}\otimes \mathbb{I}\,\Bigr)\\
   		\bra{\beta^-}_{14}=& \dfrac{1}{\sqrt 2}\Bigl(\bra{0} \otimes \mathbb{I} \otimes \mathbb{I} \otimes  \bra{1} \otimes \mathbb{I}  \otimes  \mathbb{I} - \bra{1} \otimes \mathbb{I} \otimes \mathbb{I} \otimes \bra{0}  \otimes  \mathbb{I}\otimes \mathbb{I}\,\Bigr)
   		\end{aligned}
   		\end{equation}
   	\end{scriptsize}
   	Here, $P_1$ measures say for instance the state $\bra{\alpha^+}_{14}$, then the resultant system will collapse into
   	\begin{align}
   	\resizebox{.85\hsize}{!}{$\bra{\alpha^+}_{14}\ket{\psi'}={} \left(\dfrac{\bra{0}_1 \bra{0}_4  + \bra{1}_1\bra{1}_4}{\sqrt 2}\right)\left[\left(\dfrac{\ket{1}_1\ket{0}_2+ \ket{0}_1\ket{1}_2}{\sqrt{2}}\right) \left(\dfrac{\ket{0}_3\ket{0}_4+ \ket{1}_3\ket{1}_4}{\sqrt{2}}\right)\left(\dfrac{\ket{0}_5\ket{0}_6+ \ket{1}_5\ket{1}_6}{\sqrt{2}}\right)\right]$} \\
   	=\frac{1}{2}\left( \frac{\ket{1}_2\ket{0}_3+\ket{0}_2\ket{1}_3}{\sqrt{2}}\right)\left(\frac{\ket{0}_5\ket{0}_6+ \ket{1}_5\ket{1}_6}{\sqrt{2}}\right).\end{align}
   	It can be seen from the above expression, the particles $5^{th}$ and $6^{th}$ are unaltered as expected since the measurement effects qubits $(1,40$ and $(2,3)$ only.
   	\item The state can be read as 
   	$ \ket{\chi}=\dfrac{1}{2}\ket{\beta^+}_{23}\ket{\alpha^+}_{56}$
   	. The qubits held by $P_2$ and $P_3$ are $(2,6)$ and $(3,5)$, respectively. We rearrange these qubits by swapping their places. Since they have control of their qubits, therefore swapping is done to get $\ket{\chi}$ as \begin{equation}\label{Chii}
   	\resizebox{1\hsize}{!}{$                       \ket{\chi}=\dfrac{1}{4}\left(\ket{1}_2\ket{0}_6\ket{0}_3\ket{0}_5 +\ket{1}_2\ket{1}_6\ket{0}_3\ket{1}_5+ \ket{0}_2\ket{0}_6\ket{1}_3\ket{0}_5+ \ket{0}_2\ket{1}_6\ket{1}_3\ket{1}_5    \right)$}
   	\end{equation}
   	\item Using the Bell representation of two qubits states, the state $\ket{\chi}$ can be written as
   	\begin{scriptsize}
   		\begin{equation}\label{rep}
   		\begin{aligned}
   		\ket{\chi}={}&\dfrac{1}{4\sqrt{2}} \Bigg[\left(\ket{\beta^+}_{26}-\ket{\beta^-}_{26}\right)\ket{0}_3\ket{0}_5+ \left(\ket{\alpha^+}_{26}-\ket{\beta^-}_{26}\right)\ket{0}_3\ket{1}_5\\
   		&+\left(\ket{\alpha^+}_{26}+\ket{\alpha^-}_{26}\right)\ket{1}_3\ket{0}_5  +\left(\ket{\beta^+}_{26}+\ket{\beta^-}_{26}\right)\ket{1}_3\ket{1}_5\Bigg]\\
   		&=\dfrac{1}{4}\Bigg[\ket{\beta^+}_{26}\ket{\alpha^+}_{35} -  \ket{\beta^-}_{26}\ket{\alpha^-}_{35}+ \ket{\alpha^+}_{26}\ket{\beta^+}_{35}- \ket{\alpha^-}_{26}\ket{\beta^-}_{35}\Bigg]
   		\end{aligned}
   		\end{equation}
   	\end{scriptsize}
   	\item Since $P_1$ has its disposal on the three other bell basis projective measurement operators  $\bra{\alpha^-}_{14}, \bra{\beta^+}_{14},$ and $\bra{\beta^-}_{14}$, the possible outcomes of these projective measurements collapse the system into the following three possible cases:
   	\begin{scriptsize}
   		\begin{equation}\label{eqq}
   		\begin{aligned}
   		\bra{\alpha^-}_{14}\ket{\psi'}&=\dfrac{1}{4}\Bigg[\ket{\alpha^+}_{26}\ket{\beta^-}_{35} -  \ket{\alpha^-}_{26}\ket{\beta^+}_{35}\\
   		&+ \ket{\beta^+}_{26}\ket{\alpha^-}_{35}
   		- \ket{\beta^-}_{26}\ket{\alpha^+}_{35}\Bigg]\\
   		\bra{\beta^+}_{14}\ket{\psi'}& =\dfrac{1}{4}\Bigg[\ket{\alpha^+}_{26}\ket{\alpha^+}_{35} +  \ket{\alpha^-}_{26}\ket{\alpha^-}_{35}\\
   		&+ \ket{\beta^+}_{26}\ket{\beta^+}_{35}
   		+ \ket{\beta^-}_{26}\ket{\beta^-}_{35}\Bigg]\\
   		\bra{\beta^-}_{14}\ket{\psi'}& =\dfrac{-1}{4}\Bigg[\ket{\alpha^+}_{26}\ket{\alpha^-}_{35} +  \ket{\alpha^-}_{26}\ket{\alpha^+}_{35}\\
   		&+ \ket{\beta^+}_{26}\ket{\beta^-}_{35}
   		+ \ket{\beta^-}_{26}\ket{\beta^+}_{35}\Bigg]
   		\end{aligned}
   		\end{equation}
   	\end{scriptsize}
   	\item Dealer $D$ can make three other operations $\mathbb{I}, i\mathbb{Y},\mathbb{Z}$. For each operation, the respective measurements by $P_1$ and the resultant collapsed state listed in Table \ref{Summary}.
   	\begin{table}[!htbp]
   		\begin{center}
   		\caption{Unitary operation and Corresponding Measurement Results}
   		\begin{tabular}{|c|c|c|}
   			\hline
   			D's Operation $\rightarrow$ & $P_1$'s Outcome & Collapsed State  \\ \hline
   			\multirow{ 4}{*}{$\mathbb{I}$} & $\ket{\alpha^+}_{14}$ & $\dfrac{1}{2}\ket{\alpha^+}_{23}\ket{\alpha^+}_{56}$ \\
   			& $\ket{\alpha^-}_{14}$ & $\dfrac{1}{2}\ket{\alpha^-}_{23}\ket{\alpha^+}_{56}$ \\
   			& $\ket{\beta^+}_{14}$ & $\dfrac{1}{2}\ket{\beta^+}_{23}\ket{\alpha^+}_{56}$ \\
   			& $\ket{\beta^-}_{14}$ & $ \dfrac{1}{2}\ket{\beta^-}_{23}\ket{\alpha^+}_{56}$ \\ \hline
   			\multirow{ 4}{*}{$\mathbb{X}$} & $\ket{\alpha^+}_{14}$ & $\dfrac{1}{2}\ket{\beta^+}_{23}\ket{\alpha^+}_{56}$ \\
   			& $\ket{\alpha^-}_{14}$ & $\dfrac{-1}{2}\ket{\beta^-}_{23}\ket{\alpha^+}_{56}$ \\
   			& $\ket{\beta^+}_{14}$ & $\dfrac{1}{2}\ket{\alpha^+}_{23}\ket{\alpha^+}_{56}$ \\
   			& $\ket{\beta^-}_{14}$ & $ \dfrac{-1}{2}\ket{\alpha^-}_{23}\ket{\alpha^+}_{56}$ \\\hline
   			\multirow{ 4}{*}{$i\mathbb{Y}$} & $\ket{\alpha^+}_{14}$ & $\dfrac{1}{2}\ket{\beta^-}_{23}\ket{\alpha^+}_{56}$ \\
   			& $\ket{\alpha^-}_{14}$ & $\dfrac{-1}{2}\ket{\beta^+}_{23}\ket{\alpha^+}_{56}$ \\
   			& $\ket{\beta^+}_{14}$ & $\dfrac{1}{2}\ket{\alpha^-}_{23}\ket{\alpha^+}_{56}$ \\
   			& $\ket{\beta^-}_{14}$ & $ \dfrac{-1}{2}\ket{\alpha^+}_{23}\ket{\alpha^+}_{56}$ \\\hline
   			\multirow{ 4}{*}{$\mathbb{Z}$} & $\ket{\alpha^+}_{14}$ & $\dfrac{1}{2}\ket{\alpha^-}_{23}\ket{\alpha^+}_{56}$ \\
   			& $\ket{\alpha^-}_{14}$ & $\dfrac{1}{2}\ket{\alpha^+}_{23}\ket{\alpha^+}_{56}$ \\
   			& $\ket{\beta^+}_{14}$ & $\dfrac{1}{2}\ket{\beta^-}_{23}\ket{\alpha^+}_{56}$ \\
   			& $\ket{\beta^-}_{14}$ & $ \dfrac{1}{2}\ket{\beta^+}_{23}\ket{\alpha^+}_{56}$ \\
   			\hline
   		\end{tabular}
   		\label{Summary}
   		\end{center}
   	\end{table}
   	\item The dealer $D$ announces in public (i) the product state created, (ii) qubit being transformed, i.e., $1$ or $4$.
   	\item The participant $P_1$ announces the Bell basis measurement made on his qubits $1,4$.
   	\item The participants $P_2$ and $P_3$ make their respective measurements on qubits $(2,6)$ and $(3,5)$, respectively. Hence, it colludes to find the operation done by dealer $D$ using Table \ref{Summary}.
   \end{enumerate}
   \begin{figure}[!ht]
   	\begin{center}
   		\includegraphics[width=0.80\linewidth]{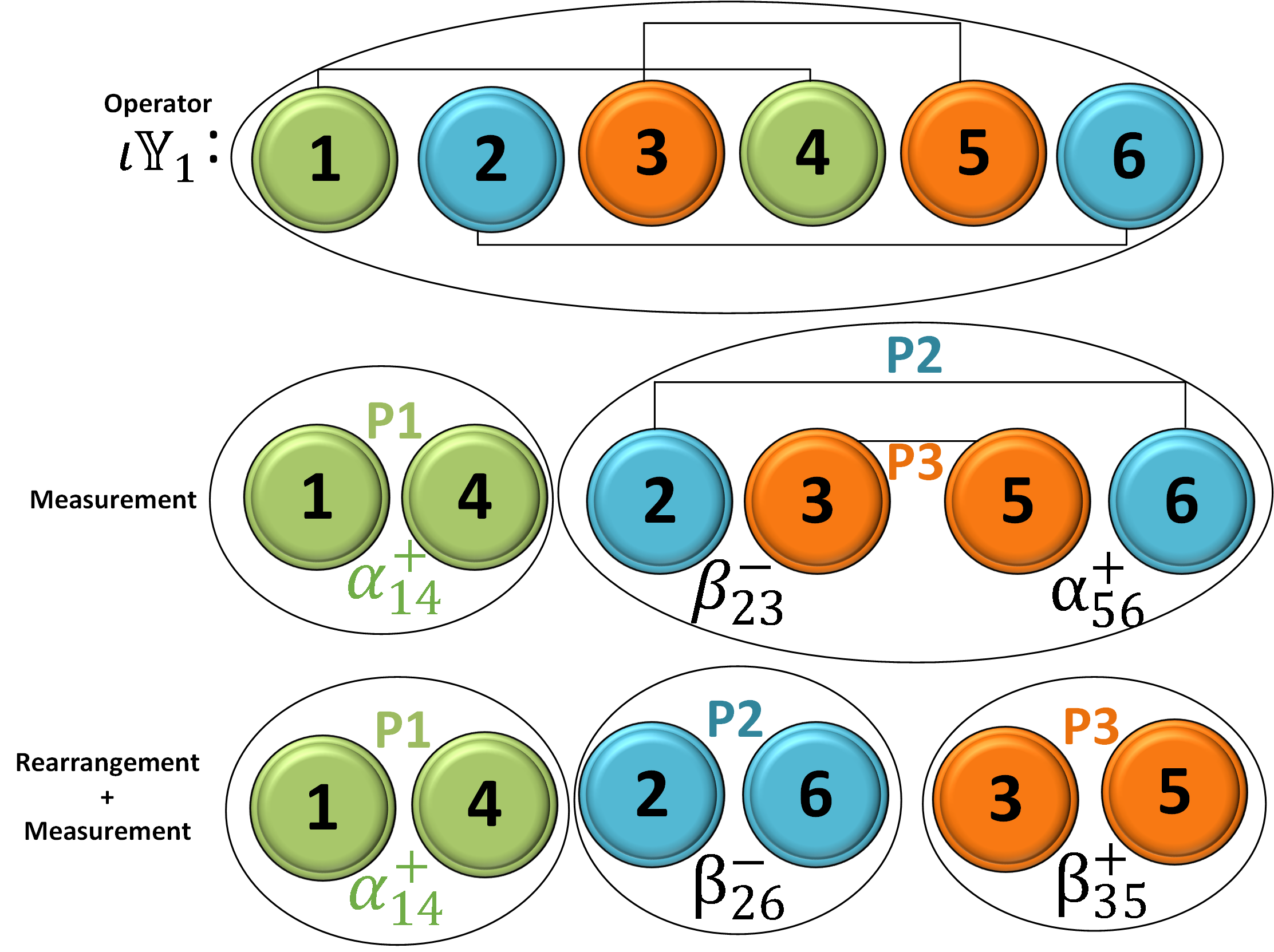}
   		\caption{Proposed protocol for quantum secret sharing}
   		\label{block}
   	\end{center}
   \end{figure}
   In this way, all the three participants hold their shares generated for the secret $s$. In the next subsection, we describe the procedure of secret reconstruction using the shares of these three participants.
   \subsection{Secret Reconstruction}
   The secret is reconstructed by the parties when each ushers in the correct information. Table \ref{Summary} gives the exact information of the unitary operation performed by the dealer $D$. We show the reconstruction result of a particular case, and a similar strategy can be adapted for other cases.  
   \begin{enumerate}
   	\item The following information are available in the beginning of the reconstruction process.
   	\begin{enumerate}
   		\item The product state created by the dealer $D$.
   		\item The qubits operated by $D$. Assume here that these qubits are $1$ or $4$.
   		\item Measurement of the participant $P_1$.
   	\end{enumerate}
   	\item Apart from the above information, the participants $P_2$ and $P_3$ announce their measurements in public, say  $\ket{\beta^-}_{26}$ and $\ket{\beta^+}_{35}$. Then the combined state becomes:
   	\begin{scriptsize}
   		\begin{equation}\label{Far}
   		\begin{aligned}\ket{\gamma}={}&\ket{\beta^-}_{26}\ket{\beta^+}_{35}\\
   		&=\left(\ket{0}_2\ket{1}_6-\ket{1}_2\ket{0}_6\right) \left(\ket{0}_3\ket{1}_5+\ket{1}_3\ket{0}_5\right)\\
   		&=\bigg[\ket{0}_2\ket{0}_3\left(\ket{\alpha^+}_{56}-\ket{\alpha^-}_{56}\right)+ \ket{0}_2\ket{1}_3\left(\ket{\beta^+}_{56}+\ket{\beta^-}_{56}\right)\\
   		&- \ket{1}_2\ket{0}_3\left(\ket{\beta^+}_{56}-\ket{\beta^-}_{56}\right)- \ket{1}_2\ket{1}_3\left(\ket{\alpha^+}_{56}+\ket{\alpha^-}_{56}\right)\bigg]\\
   		&=\ket{\alpha^-}_{23}\ket{\alpha^+}_{56} - \ket{\alpha^+}_{23}\ket{\alpha^-}_{56}+ \ket{\beta^-}_{23}\ket{\beta^+}_{56}+ \ket{\beta^+}_{23}\ket{\beta^-}_{56}
   		\end{aligned}
   		\end{equation}
   	\end{scriptsize}
   	\item Now, since they know that the qubit pair $(5,6)$ was never tampered with, hence the state of qubits $(5,6)$ should be $\ket{\alpha^+}_{56}$, which corresponds to only $\ket{\alpha^-}_{23}$ above in eqn.\eqref{Far}. Hence, the current state of their qubits is \begin{equation}\label{1}
   	\ket{\gamma}= \ket{\alpha^-}_{23}\ket{\alpha^+}_{56}
   	\end{equation}
   	\item At this step, they use the information about the measurement result of $P_1$, i.e., $\ket{\alpha^+}_{14}$. Then, write down the system as
   	\begin{scriptsize}
   		\begin{equation}\label{qq}
   		\begin{aligned}
   		\ket{\chi}={}& \left[ \left(\ket{0}_1\ket{0}_4+\ket{1}_1\ket{1}_4\right) \left(\ket{0}_2\ket{0}_3-\ket{1}_2\ket{1}_3\right)\right]\ket{\alpha^+}_{56}\\
   		&= \bigg[ \ket{0}_1\ket{0}_2\ket{0}_3\ket{0}_4- \ket{0}_1\ket{1}_2\ket{1}_3\ket{0}_4\\
   		&+ \ket{1}_1\ket{0}_2\ket{0}_3\ket{1}_4-\ket{1}_1\ket{1}_2\ket{1}_3\ket{1}_4   \bigg]\ket{\alpha^+}_{56}\\
   		&= \bigg[\ket{\alpha^+}_{12}\ket{\alpha^-}_{34}+ \ket{\alpha^-}_{12}\ket{\alpha^+}_{34}+ \ket{\beta^+}_{12}\ket{\beta^-}_{34}\\
   		&- \ket{\beta^-}_{12}\ket{\beta^+}_{34}\bigg] \ket{\alpha^+}_{56}
   		\end{aligned}
   		\end{equation}
   	\end{scriptsize}
   	\item At this point both $P_2$ and $P_3$ make use of the announcement that $D$ made about the qubit being transformed i.e., $1$ and $4$. Suppose $D$ altered qubit $1$, then they both know that qubits $(3,4)$ were not affected, i.e., they would be in state $\ket{\alpha^+}_{34}$. Using this information and eqn.\eqref{qq}, they deduce that the state is actually
   	\begin{equation}\label{aa}
   	\ket{\chi}=\ket{\alpha^-}_{12}\ket{\alpha^+}_{34} \ket{\alpha^+}_{56}
   	\end{equation}
   	Now, they make use of the information regarding the state that $D$ created, which was $\ket{a^+}_{12}\ket{a^+}_{34}\ket{a^+}_{56}$. Comparing this state with equation \eqref{aa}, they know that the operator acted upon qubit $1$ is $\mathbb{Z}$, since,
   	\begin{equation}\label{zz}
   	\mathbb{Z}_1\ket{\alpha^+}=\ket{\alpha^-}
   	\end{equation} 
   	Hence, they deduce that transformation $\mathbb{Z}$ done by $D$. According to the already agreed protocol they share the string `10', since the operator $\mathbb{Z}_1$ corresponds to `10' according to equation \eqref{secret}.
   \end{enumerate}
\section{Experimental Demonstration} 
The applicability of our secret sharing scheme is essential aspect of proposing one. To demonstrate the experimental viability of the scheme we test it on quantum simulator offered by IBM on its cloud server 'IBM-QE (Quantum Experience) \cite{IBM}. The simulators are (i)\texttt{ibmq\_qasm\_simulator\_v0.1.547} with $1024$ shots and (ii) \texttt{ibmq\_16\_melbourne} with $1024$ shots. Since the scheme is based on entanglement and swapping the qubits, we demonstrate the entanglement property between the qubits $(1,4)$, when $(2,6)$ and $(3,5)$ are measured.
\begin{description}
	\item[1.] \texttt{ibmq\_qasm\_simulator\_v0.1.547}: Fig. \ref{C1} is the quantum circuit on the `qasm simulator' with the corresponding states of the $1^{\text{st}}$ and $2^{\text{nd}}$ qubits after measurement given in Fig. \ref{C2}, as expected the state of the qubits $(1,4)$ are entangled, after each party measures its particles $(2,6)$ and $(3,5)$ in the Bell basis and announces its results. This is ascertained by the nearly equal probabilities for the states $\ket{00}$ and $\ket{11}$ for qubits $(1,4)$.
	\begin{figure}[!htbp]
		\begin{center}
			\includegraphics[width=1\linewidth]{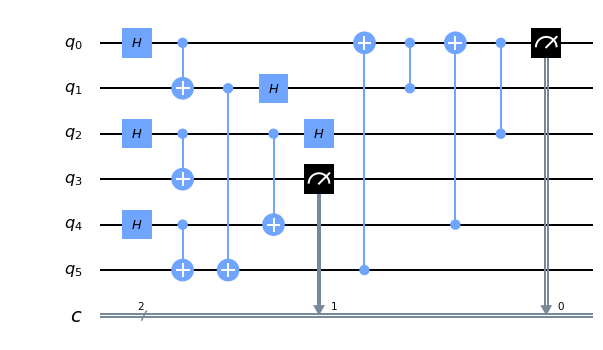}
			\caption{Quantum Circuit for 6 qubits entanglement-swapping secret sharing protocol for operator $\mathbb{I}_1$ }
			\label{C1}
		\end{center}
	\end{figure}
	\begin{figure}[!htbp]
		\begin{center}
			\includegraphics[width=0.7\linewidth]{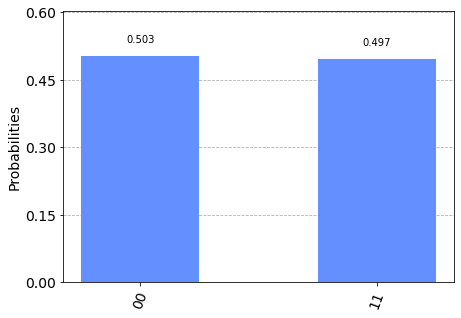}
			\caption{Probabilities for the qubits $(1,4)$ }
			\label{C2}
		\end{center}
	\end{figure}
\item[2.] \texttt{ibmq\_16\_melbourne}: The transpiled circuit on the \texttt{ibmq\_16\_melbourne} hardware is given in Fig.\ref{C3}. The corresponding measurement results on qubits $(1,4)$ are given in Fig. \ref{C4}. As can be seen from the probabilities in Fig. \ref{C4}, there are some unwanted states with non-zero probabilities, but majority are in the state that we should have theoretically. This is due to noise present and different errors in the C-NOT gates, the read out error in the actual quantum systems. The error map for the quantum hardware \texttt{ibmq\_16\_melbourne} is given in Fig. \ref{C5}. 

\begin{figure}[!htbp]
	   \centering
	\includegraphics[width=\linewidth]{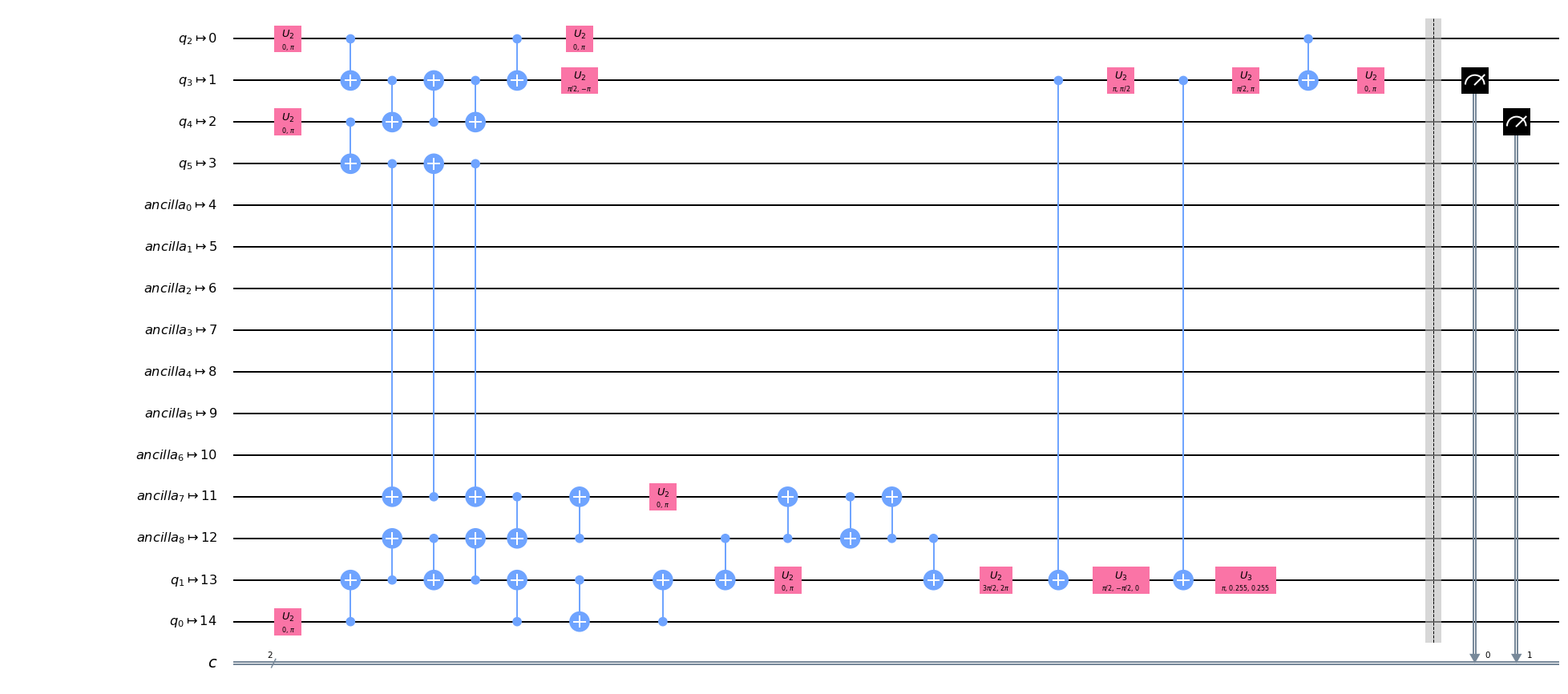}
		\caption{Quantum Circuit for 6 qubits entanglement-swapping secret sharing protocol for operator $\mathbb{I}_1$}
		\label{C3}
\end{figure}
\begin{figure}[!htbp]
	\centering
	\includegraphics[width=1\linewidth]{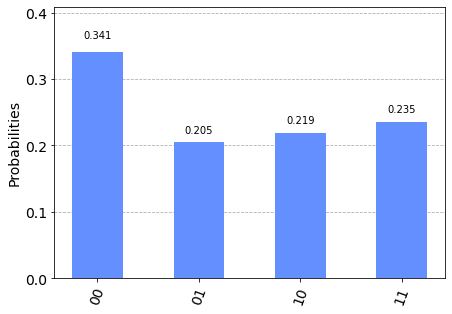}
	\caption{Probabilities for the state of qubits $(1,4)$}
	\label{C4}
\end{figure}
\begin{figure}[!htbp]
	\centering
	\includegraphics[width=1\linewidth]{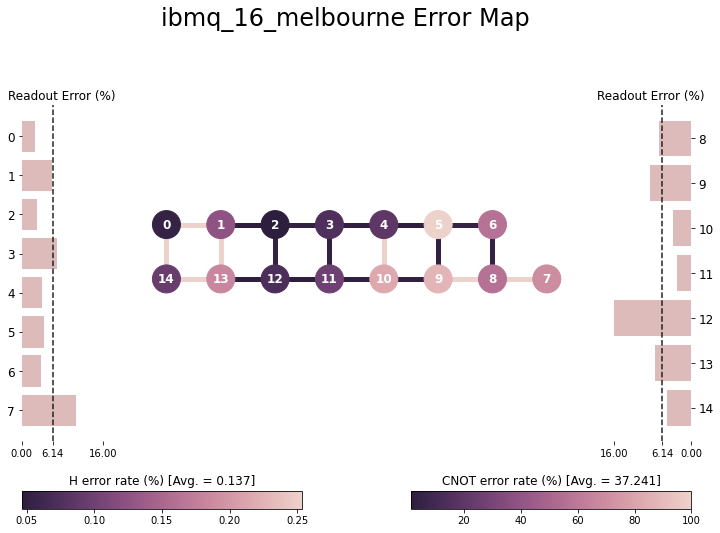}
	\caption{Error map for \texttt{ibmq\_16\_melbourne} }
	\label{C5}
\end{figure}
\end{description} 
\begin{table}[!htbp]
	\caption{ Time complexity analysis of the simulations}
	\centering
	\begin{tabular}{llll}
		\toprule
		& &	\multicolumn{2}{c}{Time} \\ 
		\cmidrule{3-4} \multicolumn{1}{l}{Backend}&\multicolumn{1}{l}{Shots}      &  \multicolumn{1}{l}{Validating}   &  \multicolumn{1}{l}{Running} \\
		\cmidrule{1-4}
		\texttt{ibmq\_qasm\_simulator\_v0.1.547}&1024&1.2s &  6ms\\
		\texttt{ibmq\_16\_melbourne}   &1024& 793ms&10s\\
		\bottomrule
	\end{tabular}
	\label{C6}
\end{table}
The time analysis of each experiment is given in Table \ref{C6}.
\section{Security Analysis}
The security analysis is an important benchmarks for any quantum secret sharing algorithm to pass in order to term it as secure and viable. To test our protocol against this benchmark, we provide proofs for the security of our protocol in terms of the participation of each party and security against an eavesdropper.
   \begin{theorem}
   	The state of each participant is maximally entangled.
   \end{theorem}
   \begin{proof}: 
   	The secrecy of the scheme depends upon many factors, of which the foremost is the lack of knowledge of each participant of his/her own subsystem. To find out what each participant sees with his particles in hand is by way of analyzing the density matrix of their subsystems. Denote the density matrix of $P_1$ is $M_{P_1}$. Since the protocol is symmetric in $P_1,P_2,$ and $P_3$, hence, other cases are equivalent. The density matrix $M$ for the entire system is
   	\begin{scriptsize}
   		\begin{equation}\label{Density}
   		\begin{aligned}
   		~~~~~M_{All}={}&\ket{\psi} \bra{\psi}\\
   		=&\frac{1}{8}\Bigg[\left(  \ket{1}_1\ket{0}_2+ \ket{0}_1\ket{1}_2\right) \left(\ket{0}_3\ket{0}_4+ \ket{1}_3\ket{1}_4\right)\left(\ket{0}_5\ket{0}_6+ \ket{1}_5\ket{1}_6\right)\\
   		&\left(  \bra{1}_1\bra{0}_2+ \bra{0}_1\bra{1}_2\right) \left(\bra{0}_3\bra{0}_4+ \bra{1}_3\bra{1}_4\right)\left(\bra{0}_5\bra{0}_6+ \bra{1}_5\bra{1}_6  \right)\Bigg]
   		\end{aligned}
   		\end{equation}
   	\end{scriptsize}
   	For $P_1$'s subsystem, we calculate 
   	$$M_{14}= \sum_a Trace_{2356}\ket{a}M_{All} \bra{a}.$$
   	To avoid this tedious calculation, we observe that the Trace effectively becomes
   	\begin{equation}\label{Trace}
   	M_{14}=Tr_{2356}=\sum_{a}\bra{a}\rho\ket{a}
   	\end{equation} where $\bra{a}= \langle0000|,\langle0100|, \langle1000|,\langle1100|,\langle0011|,\langle0111|,\langle1011|$,
   	$\langle1111|$.
   	Hence, we get \begin{equation}\label{aaa}
   	\begin{aligned}
   	M_{14}={}&\dfrac{\ket{0}\bra{0}+ \ket{1}\bra{0}+\ket{0}\bra{1}+\ket{1}\bra{1}}{4}=\dfrac{\mathbb{I}\otimes\mathbb{I}}{4}
   	\end{aligned}
   	\end{equation}
   	Since $Tr(M_{14})=1, Tr(M^2_{14})=\dfrac{1}{4}<1$, we have the state of $P_1$'s qubits as maximally entangled. This means that $P_1$'s chances of measuring any state $\ket{00},\ket{01},\ket{10},\ket{11}$ are $\dfrac{1}{4}$. Even if $P_1$ used a Bell measurement instead of the usual computational basis, it would get the same result, for instance
   	\begin{scriptsize}
   		\begin{equation}\label{prob}
   		\begin{aligned}
   		\textit{P}(\ket{\alpha^+}={}&\dfrac{1}{4}\bra{\alpha^+}\Bigg[ \left( \dfrac{\ket{\alpha^+}+\ket{\alpha^-}}{\sqrt{2}}   \right) \left( \dfrac{\bra{\alpha^+}+\bra{\alpha^-}}{\sqrt{2}}   \right)\\
   		&+ \left( \dfrac{\ket{\beta^+}-\ket{\beta^-}}{\sqrt{2}}   \right) \left( \dfrac{\bra{\beta^+}-\bra{\beta^-}}{\sqrt{2}}   \right)\\
   		&+  \left( \dfrac{\ket{\beta^+}+\ket{\beta^-}}{\sqrt{2}}   \right) \left( \dfrac{\bra{\beta^+}+\bra{\beta^-}}{\sqrt{2}}\right)\\
   		&+ \left( \dfrac{\ket{\alpha^+}-\ket{\alpha^-}}{\sqrt{2}}   \right) \left( \dfrac{\bra{\alpha^+}-\bra{\alpha^-}}{\sqrt{2}}   \right)    \bigg]\ket{\alpha^+}\\
   		&=\dfrac{1}{4}
   		\end{aligned}
   		\end{equation}
   	\end{scriptsize}
   \end{proof}
   Thus, the indication of a maximally mixed state for each of the participants ensures there is no information leakage occurs before the protocol begins, as none of them is in a position to guess the other entangled particle.
   \begin{theorem}
   	{$D$'s announcement of $\ket{\psi}$ is imperative for the correct reconstruction.}
   \end{theorem}
   \begin{proof}: The information announced by $D$ about the product state and the changed qubit are very crucial in reconstructing the secret. Suppose $D$ shares the state $\ket{\psi}=\ket{\alpha^-}_{12}\ket{\alpha^-}_{34}\ket{\alpha^-}_{56}$ and announces that he shares \begin{scriptsize}$\ket{\phi}=\ket{\alpha^+}_{12}\ket{\alpha^+}_{34}\ket{\alpha^+}_{56}$\end{scriptsize} between the participants. Suppose he wanted to share the bit-string `00', so he used the $\mathbb{I}$ operator. Suppose $P_2$ measures $\ket{\beta^-}_{26}$ and $\ket{\beta^+}_{35}$. They use $D$'s information about the shared product state $\ket{\phi}$, to obtain the state of $(2,3)$ as $\ket{\alpha^-}_{23}$. They both by the announcement of $P_1$'s measurement say $\ket{\beta^-}_{14}$, deduce that the operation is $i\mathbb{Y}$ which is the incorrect outcome. Hence,
   	\begin{scriptsize}
   		\begin{equation}
   		\begin{aligned}
   		\textit{P}\text{(Right Operator)}={}&\text{\textit{P}(Rightly inferred state of (2,3))}\cdot\text{\textit{P}(Rightly announced state of (1,4))}\\
   		&=\dfrac{1}{16}
   		\end{aligned}
   		\end{equation}
   	\end{scriptsize}
   	Since the protocol shares $4$ two-bit strings, therefore the probability of transmitting the correct bit-string is \begin{equation}\label{qrq}
   	\text{\textit{P}(Correct string)}=\left(\dfrac{1}{16}\right)^4
   	\end{equation}
   	We can see that the probability of correct decoding of the bit string tends to zero as the length of the string increases if the announcement made by $D$ is incorrect. Thus, the protocol is intricate, dependent on $D$'s announcement of the product state shared.
   \end{proof}
   \begin{theorem}
   	{$D$'s announcement of the qubit transformed is imperative for the correct reconstruction.}
   \end{theorem}
   \begin{proof}: 
   	The actual protagonist of the protocol happens to be $D$ since it is $D$ who initiates the protocol, it has to usher correct information to the colluding parties for the reconstruction of the secret. Suppose instead of the first qubit, $D$ toggled the fourth qubit and did not announce this information. Then the situation would be, even if the participants $P_1, P_2, P_3$ know about the shared product state they would not be able to reconstruct the secret. For instance, if  $P_2$ measures $\ket{\beta^-}_{26}$ and $\ket{\beta^+}_{35}$, then upon information provided by $D$ on the shared state they would come up with the choice of the state being  $\ket{\alpha^+}_{12}\ket{\alpha^-}_{34}\ket{\alpha^+}_{56}$ or $ \ket{\alpha^-}_{12}\ket{\alpha^+}_{34} \ket{\alpha^+}_{56}$. At this point, they know the operator being used, which is the $\mathbb{Z}$ operator. But, they will still not be able to reconstruct the secret message, i.e., whether it is `01' or `10' until and unless $D$ decides to give in the information about the qubit he toggled. If $D$ announces that it is the $4th$ qubit that was toggled, then the parties know that they have the state  $\ket{\alpha^+}_{12}\ket{\alpha^-}_{34}\ket{\alpha^+}_{56}$ and deduce the bit string `01' corresponding to $\mathbb{Z}_4$. Thus they are wrong half  the times, thereby $\text{Pr(Correct string)}=\left(\dfrac{1}{2}\right)^4$.
   \end{proof}
   \begin{theorem}
   	{If $P_1$ lies about his measurement, then the protocol is compromised.}
   \end{theorem}
   \begin{proof}: $P_1$ is a trusted reconstructor of the secret, and his integrity lies with utmost importance in the reconstruction of the secret bit string. However, there can be a case where he cheats and does not let his share to the other parties. Let $D$ use operator $\mathbb{Z}_1$  to send `10'. As a case of $P_1$'s cheating behaviour consider $P_1$ measuring $\ket{\alpha^+}_{14}$ on his qubits, but deliberately cheats and announces as $\ket{\alpha^-}$, $P_2$ and $P_3$ collude with their measurements $\ket{\beta^-}_{26}$ and $\ket{\beta^+}_{35}$ assuming correct information furnished by $D$ and $P_1$. They deduce the state to be \begin{equation*}
   	\begin{aligned}
   	\ket{A}={} & \ket{\alpha^+}_{12}\ket{\alpha^+}_{34}\ket{\alpha^+}_{56}\\
   	&\implies  \mathbb{I}_1\\
   	& \implies  `00'
   	\end{aligned}
   	\end{equation*}
   	whereas the actual state would have been
   	\begin{equation*}
   	\begin{aligned}
   	\ket{A}={} & \ket{\alpha^-}_{12}\ket{\alpha^+}_{34}\ket{\alpha^+}_{56}\\
   	&\implies  \mathbb{Z}_1\\
   	& \implies  `10'
   	\end{aligned}
   	\end{equation*}
   	Hence, we infer that the announcement of $P_1$'s announcement of his measurement outcome is necessary, failing which, the correct operator and hence the correct classical secret will not be reconstructed.
   \end{proof}
   \begin{lemma}
   	{If either one of $P_1$ or $P_2$ lies, then the protocol is compromised.}
   \end{lemma}
   \begin{proof}: The roles of $P_2$ and $P_3$ are symmetric. Suppose that $P_2$ does not reveal his measurement outcome, then from Table 1, it is quite evident that any operator is equally likely. Suppose $P_3$ measures $\ket{\beta^-}_{26}$, but $P_2$ does not disclose his measurement outcome. So in spite of $P_3, $ having $P_1$'s measurement result say $\ket{\alpha^+}_{14}$, cannot deduce the state of the qubits $(2,3)$ from eq. \eqref{Far}. Hence, they will not be able to proceed further with the secret reconstruction.
   \end{proof}
   \begin{theorem}
   	{The algorithm is secure against Eavesdropper Eve's forgery attack}
   \end{theorem}
   \begin{proof}
   	Suppose there is a malicious observer Eve who has some ill intentions about the protocol and seeks to disrupt the protocol. He somehow catches hold of a qubit that $D$ creates say the $5^{th}$ qubit and modifies it by applying the $X$ operator without anybody knowing of this act. So the modified state is
   	\begin{scriptsize}\begin{equation*}
   		\ket{\psi}_{mod}=\left(\frac{\ket{0}_1\ket{0}_2+ \ket{1}_1\ket{1}_2}{\sqrt{2}}\right) \left(\frac{\ket{0}_3\ket{0}_4+ \ket{1}_3\ket{1}_4}{\sqrt{2}}\right)\left(\frac{\ket{1}_5\ket{0}_6+ \ket{0}_5\ket{1}_6}{\sqrt{2}}\right)
   		\end{equation*}\end{scriptsize} $P_1$ measures say $\ket{\alpha^+}_{14}$ and the state collapses into the state \begin{scriptsize}
   		\begin{equation*}\begin{aligned}
   		\ket{\alpha^+}_{23}\ket{\beta^+}_{56}={}&
   		\ket{\alpha^+}_{35}\ket{\beta^+}_{26}+ \ket{\alpha^-}_{35}\ket{\beta^-}_{26}+ \ket{\beta^+}_{35}\ket{\alpha^+}_{26}\\
   		&+ \ket{\beta^-}_{35}\ket{\alpha^-}_{26}
   		\end{aligned}
   		\end{equation*}\end{scriptsize}
   	$P_2$ and $P_3$ measure $\ket{\alpha^-}_{35}$ and $\ket{\beta^-}_{26}$ respectively. They both collude to form the secret with the help of the relation \begin{scriptsize}\begin{equation*}
   		\begin{aligned}
   		\ket{\alpha^-}_{35}\ket{\beta^-}_{26}={}&\ket{\alpha^+}_{23}\ket{\beta^+}_{56}+ \ket{\alpha^-}_{23}\ket{\beta^-}_{56} - \ket{\beta^+}_{23}\ket{\alpha^+}_{56}\\
   		&+ \ket{\beta^-}_{23}\ket{\alpha^-}_{56}
   		\end{aligned}
   		\end{equation*}\end{scriptsize}
   	They make use of $D$'s announcement of the product state he created which was $\ket{\alpha^+}_{12}\ket{\alpha^+}_{34}\ket{\alpha^+}_{56}$, which corresponds to the third term in the above equation. They use $P_1$'s  measurement of $\ket{\alpha^+}_{14}$ to deduce that the operator was actually $\mathbb{I}$, thereby nullifying Eve's attack.
   \end{proof}
   \section{Comparative Analysis}
   This section presents an extensive comparative analysis with the recent developed quantum secret sharing algorithms.
   \begin{enumerate}
   	\item The work proposed in \cite{song2017t} asserts of a secure $(t,n)$ quantum secret sharing scheme that is based on Quantum Fourier Transform (QFT) and Pauli operators for secret sharing. Their claim for the secret reconstruction seems to be on spurious grounds, where they apply the IQFT to a single qubit of an entangled system in isolation for reconstruction. The equation they use for reconstruction is \begin{equation}\label{QFTT}
   	\text{Secret}=\text{IQFT}\left(\dfrac{1}{\sqrt{d}}\sum_{k=0}^{d-1}\omega^{\left(\sum_{r=1}^{t}s_rk\right)}\ket{k}_1\ket{k}_2...\ket{k}_t\right)
   	\end{equation}
   	The algorithm applies the IQFT on the first qubit alone to get the secret as $s_r$ which is definitely absurd. Since the actual implementation would give $\left(\dfrac{1}{d}\sum_{k=0}^{d-1}\sum_{x=0}^{d-1}\omega^{\left(\sum_{r=1}^{t}s_r+kx\right)}\ket{x}\ket{k}_1\ket{k}_2...\ket{k}_t\right)\neq \ket{\sum_{r=1}^{t}s_r \mod d}$ as claimed. the claim would have been right if the QFT is applied to the individual 1st qubit in isolation, i.e.,
   	\begin{scriptsize}\begin{equation}
   		\text{IQFT}\left(\dfrac{1}{d}\sum_{k=0}^{d-1}\omega^{(\sum_{r=1}^{t}s_rk)}\ket{k}_1\right)= \sum_{k=0}^{d-1}\omega^{(\sum_{r=1}^{t}s_rk)}\sum_{j=0}^{d-1}\omega^{-jk}\ket{j}\ket{k}_1
   		\end{equation}\end{scriptsize}
   	The protocol thus does not pass the test of reconstruction as claimed, and can be termed as a compromised one. Whereas our proposed scheme supports safe and secure transmission of the qubits and facilitates an $(n,n)$ reconstruction of the desired secret by means of an entangle-measure-announce scheme adopted by the participants.
   	\item    One of the recent works reported in \cite{mashhadi2019general} work on an amalgamation of classical and quantum secret sharing and at the core utilizes the monotone span program from the classical secret sharing scheme to generate shares. The basic drawback in the scheme that we could notice is (i) the secret $s<d$, where $d$ is the dimension of the quantum system, whereas our scheme relies on operators as basic communication sources rather than the secret value which is implicitly agreed upon by the participants. Secondly, the actual security lies in the classical sharing scheme and is scarcely related to the quantum aspect of sharing. As soon as an eves-dropper or even a dishonest participant gets hold of the classical scheme, he doesn't even need to do the measurement of the qubits to know the secret, without invoking the basics of quantum secret sharing. Most of the quantum secret sharing schemes proposed rely upon the Quantum Fourier Transform (QFT) for secret sharing and secret reconstrcution which includes protocols reported in \cite{yang2013secret,qin2018multi}.
   	Most of the quantum secret sharing schemes rely upon the Quantum Fourier Transform (QFT) for secret sharing and secret reconstruction, which includes protocols reported in \cite{yang2013secret,qin2018multi}.
   	\item The work reported in \cite{qin2016verifiable} is a more quantum scheme than the others in the sense that it realizes the Bell state measurement for the secret reconstruction and verification. The major hindrance to the applicability of the scheme we inferred is again the strict assumption that the secret $s<d$.
   \end{enumerate}
   
  \begin{table}[htbp]
  	\refstepcounter{table}
 	\rotatebox[origin=c]{90}{\textbf{Table}~\thetable: Comparisons with schemes}
  \centering\small\setlength\tabcolsep{2pt}
  \resizebox{0.5\textwidth}{!}{%
  		\rotatebox{90}{
  			\begin{tabular}{l l l l l l l l l l l l l c c c c c c c c c c c  }
  				\toprule
  				\textbf{Parameters} &  Song et al.\cite{song2017t}    &Qin et al.\cite{qin2018multi}  &   Qin and Dai \cite{qin2016d} & Yang et al. \cite{yang2013secret} & Mashhadi\cite{mashhadi2019general} & \textbf{Proposed}\\
  				\midrule
  				Secret Reconstruction       &  \xmark   & \checkmark & \checkmark & \checkmark & \checkmark & \checkmark\\
  				Security                     &  Not secure   & Secure & Not secure   & Secure & Secure & Secure\\
  				\small{Secret type}           &   Classical &    Classical  &  Classical &  Classical  &  Classical &  Quantum\\
  				Basic scheme          &   QFT + CNOT gate       &   QFT+Pauli Gate  &  Pauli Gate &   QFT+Pauli Gate & QFT+Pauli Gate & Entanglement \\
  				Classical secret sharing & \checkmark & \checkmark      & \checkmark            & \checkmark    & \checkmark         & \xmark\\
  				\small{Dimensionality of the system} &  $d $ &  $d$   & $d$  &  $d$  & $d$ & $2$ \\
  				Technique used               &  Multi-party scheme  & Multiparty  & Access structure  & Multiparty & Multiparty & Sequential Measurement\\
  				\small{Eves dropper Manipulations}           &  Possible  &     Not possible  &  Possible  & Not possible &  Not possible & Not possible\\
  				Qubits used                  & $\lceil \log_2d\rceil n$  &  $\lceil \log_2d\rceil n$   & $ \log_2d$  &  $\lceil \log_2d\rceil n$ & $\lceil \log_2d\rceil m$  &  $2n$ \\
  				Entanglement used               & \xmark  &    \xmark   & \xmark   & \xmark   & \xmark  &\checkmark\\
  				\small{Measurement effect on others}    & \xmark   & \xmark  &   \xmark  &   \xmark &   \xmark & \checkmark  \\
  				\small{No. of Operations}        &  $5+t$ &    $t(t+1)+2n+2$  & $n+1$   & $3n$  &  $2m+3$  & 1\\
  				\small{No. of Measurements}      & $t$     &     $t$           & $n$  &     $n$   &    $m$   & $3$\\
  				\bottomrule
  	\end{tabular}}}
  \label{tab:gr}
  \end{table}
   \section{Application}
   The security of digital images is a key challenge faced by many big data analytics companies and multinationals. The digital image may be of varied importance ranging from military maps to digital signatures. The proposed method can provide perfect visual secret sharing by harnessing the Novel Enhanced Quantum Representation (NEQR)\cite{zhang2013neqr}. The image is represented by
   \begin{equation}\label{NEQR}
   \begin{aligned}
   \ket{A}&=\dfrac{1}{2^n}\sum_{i=0}^{2^n-1}\sum_{j=0}^{2^n-1}|C(i,j)\rangle\otimes |i\rangle |j\rangle\\ &=\dfrac{1}{2^n}\sum_{i=0}^{2^n-1}\sum_{j=0}^{2^n-1}|C^0(i,j)C^1(i,j)...C^{q-1}(i,j)\rangle\otimes |i\rangle |j\rangle
   \end{aligned}
   \end{equation}
   where, $C^x(i,j)\in\{0,1\}~ \forall i,j,x$. For a gray level image $C(i,j)\in \{0,1,.....255\}$, we have $q=8$. Hence, in the NEQR model, to represent a $q$ bit depth image of size of size $2^n\times 2^n$, we need a total of $q+2n$ qubits. Compactly, it can be written as \begin{equation}\label{NEQR1}
   |A\rangle=\dfrac{1}{2^n}\sum_{i=0}^{2^n-1}\sum_{j=0}^{2^n-1} \Motimes\limits_{x=0}^{q-1}|C^x(i,j)\rangle\otimes |ij\rangle
   \end{equation}
   \begin{enumerate}
   	\item To set the value of each pixel in the Quantum states by defining an operator $\tilde{U}$, as follows : \begin{equation} \tilde{U}_{xy}=I \otimes \sum_{i=0}^{2^n-1}\sum_{\substack{j=0\\ ij\neq xy}}^{2^n-1}|ij\rangle \langle ij|+\tilde{Set}_{xy}\otimes |xy\rangle \langle xy|
   	\end{equation}
   	where, $\displaystyle \tilde{Set}_{xy}|0^q\rangle~\mapsto |0\oplus \Motimes\limits_{u=0}^{q-1}\tilde{A}^u(x,y)\rangle =\Motimes _{u=0}^{q-1}|0\oplus \tilde{A}^u(x,y)\rangle$, where $\tilde{A}^u(x,y)\in \{0,1\} \forall u,x,y$.
   	\item Applying $\tilde{U}_{xy}$ on $\ket{\Phi_2}$ we set the pixel intensity for each coordinate $(x,y)$ by the equation \begin{equation}\label{NN}
   	\begin{aligned}
   	\tilde{U}_{xy}\ket{\Phi_2}&= \tilde{U}_{xy}\left(  \sum_{i=0}^{2^n-1}\sum_{\substack{j=0\\ ij\neq xy}}^{2^n-1}\ket{0}^q\otimes |ij\rangle +\ket{0}^q\otimes |xy\rangle\right)\\
   	&= \sum_{i=0}^{2^n-1}\sum_{\substack{j=0\\ ij\neq xy}}^{2^n-1}\ket{0}^q\otimes |ij\rangle +\tilde{A}(x,y)\otimes |xy\rangle
   	\end{aligned}
   	\end{equation}
   	\item The operation of $\tilde{U}_{xy}$ on the state $\ket{\Phi_2}$ sets the value of a single pixel. Therefore,the application of the operator is to be done $2^{3l}$ times to set the value for each pixel. Hence, we define an operator $\mathbf{U}$\begin{equation*}
   	\mathbf{\tilde{U}}=\prod_{i=0}^{2^n-1}\prod_{j=0}^{2^n-1}\tilde{U}_{ij}
   	\end{equation*} which acts on the state $\ket{\Phi_2}$ to  set all the values of the image and obtain the final state $\ket{\Phi_3}$.
   	\begin{equation*}\label{AA}
   	\begin{aligned}
   	\ket{\Phi_3}={}&\mathbf{\tilde{U}}\ket{\Phi_2} = \prod_{i=0}^{2^n-1}\prod_{j=0}^{2^n-1}\tilde{U}_{ij}\ket{\tilde{A_2}}\\
   	\ket{\Phi_3}=&\frac{1}{2^n}\sum_{x=0}^{2^n-1}\sum_{y=0}^{2^n-1}|\tilde{A}^0(x,y)...\tilde{A}^{q-1}(x,y) \rangle\otimes |xy\rangle\\
   	=&\frac{1}{2^n}\sum_{x=0}^{2^n-1}\sum_{j=0}^{2^n-1}|\tilde{A}(x,y)\rangle\otimes |xy\rangle
   	\end{aligned}
   	\end{equation*}where, $x=|x_0x_1....x_{n-1}\rangle, y=|y_0,y_1,....y_{n-1}\rangle$ and $\tilde{A}^t(x,y)\in\{0,1\}~\forall x,y,t$.
   \end{enumerate}
   Example 1. The dealer $D$ wishes to share the pixel at $(1,1)$ with intensity $55$, so he decides to share the state $\ket{110111}\ket{01}\ket{01}$. The state $\ket{110111}=i\mathbb{Y}_1\mathbb{X}_1\mathbb{I}_4, \ket{01}=\mathbb{Z}_4 \text{or} \mathbb{X}_1$.
   \section{Conclusion}
   The novelty of the proposed scheme is based on the reliance of pure quantum mechanical properties of photons as compared to a layer of classical cryptography implemented in other schemes. The security of the scheme against an adversary is quite substantial, and also malicious alterations to the secret is not possible. The scheme stands the test of various scenarios of data manipulations and partial information availability. The applicability of the proposed scheme for visual secret sharing is given that enhances the viability of the scheme. The prospects for the scheme are its generalization for security purposes using the $GHZ$ and $W$ state to devise a secure quantum secret sharing algorithm. The limitation that we find in our scheme is the cascading effect of the errors involved during the transmission of the quantum secret over long distances and noisy channels. We adhere to improve this concern by various quantum error mitigation techniques.
   \section*{Acknowledgement}
   One of the authors, Farhan Musanna, with grant number MHR-01-23-200-428 is grateful to Ministry of Human Resource Development (MHRD), Government of India and  Indian Institute of Technology Roorkee, for providing financial aid for this work. The authors are extremely thankful to IBM for providing access to their Quantum Experience (IBM-QE) cloud servers.

   \end{document}